\documentclass{lmcs} 


\usepackage{hyperref}
\theoremstyle{plain} 

\begin{document}

\title[An extended description logic system with knowledge element based on ALC]{An extended description logic system with knowledge element based on ALC \texorpdfstring{\MakeLowercase{\texttt{}}}{}\rsuper*\\}
\titlecomment{{\lsuper*} Partially supported by National Nature Science Fund Project (No. 61661051 and 61562093), Key Project of Applied Basic Research Program of Yunnan Province (No. 2016FA024).}

\author[Bin Wen]{Bin Wen$^{1}$}	
\address{\small 1.School of Information Science and Technology, Yunnan Normal University, Kunming 650500, China }	
\email{wenbin@ynnu.edu.cn1}  

\author[Jianhou Gan]{Jianhou Gan$^{2}$}	
\address{\small 2.Key Laboratory of Educational Informatization for Nationalities Ministry of Education,Yunnan Normal University,Kunming 650500, China}	

\author[Juan L.G. Guirao]{Juan L.G. Guirao$^{3}$}	
\address{\small 3.Departamento de Matem\'{a}tica Aplicada y Estad\'{\i}stica.Universidad Polit\'{e}cnica de Cartagena,Hospital de Marina, 30203-Cartagena, Regi\'{o}n de Murcia, Spain}	

\author[Wei Gao]{Wei Gao$^{1}$}	
\address{\small 1.School of Information Science and Technology, Yunnan Normal University, Kunming 650500, China}	




\begin{abstract}
  With the rise of knowledge management and knowledge economy, the knowledge elements that directly link and embody the knowledge system have become the research focus and hotspot in certain areas. The existing knowledge element representation methods are  limited in functions to deal with the formality, logic and reasoning. Based on description logic ALC and the common knowledge element model, in order to describe the knowledge element, the description logic ALC is expanded. The concept is extended to two different ones (that is, the object knowledge element concept and the attribute knowledge element concept). The relationship is extended to three (that is, relationship between object knowledge element concept and attribute knowledge element concept, relationship among object knowledge element concepts, relationship among attribute knowledge element concepts), and the inverse relationship constructor is added to propose a description logic KEDL system. By demonstrating, the relevant properties, such as completeness, reliability, of the described logic system KEDL are obtained. Finally, it is verified by the example that the description logic KEDL system has strong knowledge element description ability.
\end{abstract}

\maketitle

\section*{Introduction}\label{S:one}

  With the rise of knowledge management and knowledge economy, knowledge has attracted people's attention as an important strategic resource. The direct control and management of knowledge itself has become the focus of attention in various disciplines. As a direct contact and knowledge unit of knowledge system, knowledge elements has become the research focus in certain fields \cite{1}. The study of knowledge elements should be traced back to the concept of data elements proposed by American intelligence scientist Vladimir Slamek in the late 1970s (see \cite{2} for more details). He believed that the knowledge units in the literature can be further decomposed into the smallest independent data elements such as data, formulas, facts, conclusions, etc., in order to better realize knowledge retrieval and knowledge management, and also help to improve the efficiency and precision of knowledge acquisition. Subsequently, domestic and foreign scholars began to conduct in-depth research on knowledge and its applications.

In recent years, the theoretical research on knowledge by scholars at home and abroad has focused on the definition of the concept of knowledge element, the type of knowledge element structure, and the representation of knowledge elements. In terms of the type of knowledge element structure, Wang  constructed a common knowledge element model. He used the system theory method to divide the knowledge element into object knowledge element, attribute knowledge element and relational knowledge element, and abstracted various concrete things in the objective world. The knowledge element model is used in emergency analysis, emergency decision making and management, video content presentation and retrieval, and enterprise knowledge management \cite{3}. In terms of knowledge element representation, the representation of knowledge elements refers to the normalization and description of knowledge elements through symbols, frameworks and models, so that the knowledge elements are symbolized, formalized or modelled for computer processing. Nowadays, the researchers generally use framework-based representations and ontology-based methods to represent knowledge elements. Zhou  et al. \cite{4} used Framework Representation The knowledge element representation method on the XML platform, and the Bacchus paradigm format of the knowledge element frame. Yuan  et al. \cite{5} proposed a unified knowledge representation method based on knowledge element ontology, in which the knowledge element ontology contains five classes: creator, knowledge element, knowledge element abstract, knowledge element description, and history. The knowledge element was used to describe different knowledge units, the knowledge element description used to represent the knowledge element description body, and the history used to record the evolution of knowledge elements. Lu R-Q and others from the Chinese Academy of Sciences believed that knowledge is the basic unit of knowledge representation in ontology, and a large collection of knowledge elements become the material of knowledge engineering \cite{6}. In the representation of knowledge element structure, Gao  et al. \cite{7} put forward the concept of knowledge element, relationship, question triplet representation. Jiang \cite{8} proposed knowledge element name, attribute, attribute description and information interface, and the tuple was represented. Ma \cite{9} proposed the knowledge element concept, knowledge element attribute, knowledge element method and knowledge element relationship group representation. Liu and Wang \cite{10} proposed a six-tuple based (number, navigation, source, type, feature words, content) A method to express knowledge elements. Wen and Wen \cite{11} proposed seven elements: knowledge element name, representation, rule information, operation information, navigation information, superior information and related information. From the perspective of existing knowledge element representation methods, it is not strong in terms of formality, logic and reasoning.

Description Logic (DL) is a formal method based on object-based knowledge representation. It draws on the main idea of KL-ONE and is a decidable subset of the first-order predicate logic. It differs from the first-order predicate logic in that the description logic system can provide decidable reasoning services. An important feature of the description logic is its strong expressiveness and decidability, which ensures that the inference algorithm always stops and returns the correct result. The description logic has a clear model-theoretical mechanism; it is well suited to represent the application domain through conceptual taxonomy and provide useful reasoning services. The traditional description logic ALC with intersection ($\sqcap$), union ($\sqcup$), complement ($\neg$), existing quantifier ($\exists $) and full quantifier ($\forall$) constructor, has a certain potential to be described. As a formal aspect, description logic is based on objects, for knowledge representation has been used in information systems, language understanding, software engineering (see Baader et al. \cite{12}), Semantics Web (see Baader et al. \cite{13}), Knowledge representation (see Calvanese et al. \cite{14}), and Bioinformatics integration (see Schulz \cite{15}) etc. In addition, Niu et al. \cite{16} proposed a description logic based on the knowledge description method in the field of motor fault diagnosis. The logic fault detection reasoning of the built motor fault knowledge base is carried out. To enhance the ability of expressing uncertainty of the knowledge description logic, Chen  et al. \cite{17} proposed a method to extend the uncertainty of the description logic SROIQ(D). Fang et al. \cite{18} determined a new description logic based on the theory of dynamic fuzzy logic, that is Dynamic Fuzzy Description Logic (DFDLs), and used DFDLs to represent the uncertain knowledge of Deep Web and achieved reasonable reasoning and utilization. Based on the common knowledge element model proposed by Wang Yanzhang and description logic ALC, in order to describe the knowledge element, we extend the description logic ALC. The concept is extended to two different ones(that is, object knowledge element concept and attribute knowledge element concept). The relationship is extended to three (that is, relationship between object knowledge element concept and attribute knowledge element concept, relationship among object knowledge element concepts, relationship among attribute knowledge element concepts), and the inverse relationship constructor is added to propose a description logic KEDL system. By demonstrating, the relevant properties, such as completeness, reliability, of the description logic system KEDL are obtained.

\section{Knowledge element model}

It is well known that the knowledge elements are applied to the model domain and gave a common knowledge element model, namely the model knowledge element (\cite{19}, \cite{20}, and \cite{21}). It is stated below:\\
(1) Object knowledge element.

For a model \textit{m}, let $N_{m}$ be the concept of things represented by \textit{m} and its attribute names,  $A_{m}$ is the relevant attribute state set of \textit{m}, $R_{m}$ is the mapping set on $A_{m}\times A_{m}$. Then the knowledge element corresponding to $m$ is:
\begin{equation}\label{1}
K_{m}=(N_{m},A_{m},R_{m}).
\end{equation}
Generally, there are $N_{m}\neq\emptyset$,$A_{m}\neq\emptyset$,$R_{m}\neq\emptyset$.

(2) Attribute knowledge element.

Assume $a\in A_{m}$, its corresponding knowledge element is
\begin{equation}\label{2}
K_{a}=(p_{a},d_{a},f_{a})\quad{\rm for}\quad \forall a \in A_{m},
\end{equation}
where $p_{a}$ is the measurable feature description the attribute state $a$. When $p_{a}=0$, the attribute status $a$ is non-descriptive. When $p_{a}=1$, the attribute status $a$ is descriptive. If $p_{a}=2$, the attribute status $a$ is conventional measurable. If $p_{a}=3$, the attribute status $a$ is random measurable. If $p_{a}=4$, the attribute status $a$ is fuzzy measurable, etc. If $a$ is measurable, then $d_{a}$ is the measure dimension and $f_{a}$ is the change function of $a$£¬there is $p_{a}>0$,$d_{a}\neq\emptyset$, but $f_{a}$ may be empty.

(3) Relational knowledge element.

Assume $r\in R_{m}$, its corresponding knowledge element is
\begin{equation}\label{3}
K_{r}=(p_{r},A^{I}_{r},A^{O}_{r},f_{r})\quad{\rm for}\quad \forall r \in R_{m},
\end{equation}
where $p_{r}$ represents the mapping attribute description of the relationship $r$, which can be logical structure, linear, fuzzy, non-linear, etc. $A^{I}_{r}$ indicates the input attribute state set, $A^{O}_{r}$  indicates the output attribute state set, $f_{r}$ is mapping function, there is
$A^{O}_{r}$=$f_{r}$($A^{I}_{r}$), $p_{r}\neq\emptyset$, $A^{O}_{r}\neq\emptyset$, $A^{I}_{r}\neq\emptyset$ and $f_{r}\neq0$.

Based on the knowledge element model, researchers have applied the research into practice. Chen et al. \cite{22} refined this knowledge element model, and also pointed out that this model characterizes the individual elements and associated methods of the objective things system from three perspectives: the knowledge element itself, the knowledge element attribute and the relationship between the knowledge element attributes. They also proposed the implicit description method of the knowledge element attribute relationship and proved that the knowledge element model is valid with examples. Zhong et al. \cite{23} studied the scenario-related models in emergencies, and established and described element models, conceptual models, and instantiation constraints based on knowledge elements. Based on the above knowledge element model,  Zhong et al. \cite{24} instantiated in the hazard system to obtain the hazard body knowledge element, and gave its model. Chao \cite{25} proposed that the knowledge element can be regarded as the result of the concreteization of the common knowledge element model. Based on the knowledge element, we study the chain reaction in emergencies from the knowledge level. Wang et al. \cite{26} introduced the common knowledge element model to extract common knowledge of emergency management in the field of emergency events, to form a common and scalable emergency knowledge element system, and proposed a situational representation of the emergency management case and its storage model. Yu et al. \cite{27} used knowledge elements to describe the common features of the events, and constructed an emergency risk prediction model based on knowledge elements. Zhang et al. \cite{28} proposed a scenario representation model based on the knowledge element model and the micro-analysis of scenarios factors and their interactions. In view of the uncertainty of information in scenarios inference, a fuzzy inference method was introduced. Sun et al. \cite{29} constructed a knowledge system for competitive intelligence based on the knowledge element, then generated a network with the help of relationship among the attributes of these meta-data, finally identified competitive intelligence through similarity analysis and merging multi-attributes.

\section{Extended description logic KEDL}

\subsection{Propose description logic KEDL}

The traditional description logic ALC is composed of intersection ($\sqcap$), union ($\sqcup$), complement ($\neg$), exist quantifier ($\exists$) and full quantifier ($\forall$) constructor. It has a certain ability to describe knowledge. Using description logic ALC to describe knowledge elements, instance of concept is the instance of the object knowledge element, that is, only when the concept C is interpreted as a set of object knowledge element instances and a is interpreted as an instance of the object knowledge element, $C(a)$ is meaningful in description logic ALC. There is now a fact that there are two concepts in the concept description with knowledge element: (1) The concept whose instance is the thing object. (2) The concept whose instance is the attribute value. For example, in Chen et al. \cite{22}, based on the knowledge element system, the knowledge element building was extracted from the debris flow disaster, which is described as: building {location, state, function, number of people, impact resistance, shape factor, impact force}, according to the case after instantiation in the literature, instances of object knowledge element buildings include buildings 1, buildings 2, instances of attribute locations include (120,40), (120,40,1), the instances of the attribute state include normal, damage, be damaged, instances of attribute impact resistance include $1000N/m^{2}$, $1200N/m^{2}$. The other examples are not listed one by one. In addition, we define a building concept with impact resistance less than $1200N/m^{2}$ (object knowledge element concept), all instances of this concept have an impact resistance property of less than $1200 N/m^{2}$, but who is the impact resistance of less than $1200N/m^{2}$ is a concept in the field of attribute impact resistance (attribute knowledge element concept). In addition, if the gas accident has the attribute of the number of casualties, the relationship between the knowledge element gas accident and the attribute knowledge element of casualty can be expressed as a relationship existing. In turn, the relationship between the attribute knowledge element casualty and the knowledge element gas accident is expressed as a relationship existed. The relationship between existing and existed is inverse relation each other. In order to be able to describe the knowledge element, we extend the description logic ALC, in which the concept is divided into two types of concepts by a kind of concept extension, the relationship is divided into three types, and the inverse relationship constructor is added to propose the description logic KEDL. The syntax and semantics for description logic KEDL are defined as follows and the relevant properties are elaborated.

\subsection{The syntax and semantics of description logic KEDL}
In description logic KEDL, let \textrm{\textbf{R}} be the relationship set between the object knowledge element and the attribute knowledge element. \textrm{\textbf{P}} represents the relationship set between object knowledge elements. \textrm{\textbf{Q}} represents the relationship set between attribute knowledge elements, \textbf{$\Phi$} represents the object knowledge element concept set, and \textbf{$\Omega$}  represents the attribute knowledge element concept set. We define the relationship between the object knowledge element and the attribute knowledge element, the relationship among the object knowledge elements, the relationship among the attribute knowledge elements, the object knowledge element concept, and the attribute knowledge element concept to satisfy the minimum set of definition 1. The syntax definition of description logic KEDL is given as follows.

\begin{defi} {\rm (Syntax)} The syntax of KEDL is stated as follows:\\
{\rm(1)} if $C, D\in\textbf{$\Phi$}$, then $\top, \bot, \neg$C$, $C$\sqcap$D$, $C$\sqcup$D$\in\textbf{$\Phi$}$.\\
{\rm(2)} if $p\in\textrm{\textbf{P}}, C\in\textbf{$\Phi$}$, then $\forall$p.C$, \exists$p.C$\in\textbf{$\Phi$}$.\\
{\rm(3)} if $r\in\textrm{\textbf{R}}, A\in\textbf{$\Omega$}$, then $\forall$ r.A$,\exists$r.A$\in\textbf{$\Omega$}$.\\
{\rm(4)} if $A,B\in\textbf{$\Omega$}$, then $\neg$A$, $A$\sqcap$B$, $A$\sqcup$B$\in\textbf{$\Omega$}$.\\
{\rm(5)} if $q\in\textrm{\textbf{Q}}, A\in\textbf{$\Omega$}$, then $\forall$q.A$,\exists$q.A$\in\textbf{$\Omega$}$.\\
{\rm(6)} if $r\in\textrm{\textbf{R}}, C\in\textbf{$\Omega$}$, then $\forall$r-.C$,\exists$r-.C$\in\textbf{$\Omega$}$.\\
{\rm(7)} if $r\in\textrm{\textbf{R}}$, then $r^{-}\in\textrm{\textbf{R}}$.\\
{\rm(8)} if $C,D\in\textbf{$\Phi$}$, then $C\rightarrow$D$, $C$\leftrightarrow$D$\in\textbf{$\Phi$}$.\\
{\rm(9)} if $A,B\in\textbf{$\Omega$}$, then $A\rightarrow$B$, $A$\leftrightarrow$B$\in\textbf{$\Omega$}$.\\
\end{defi}

\begin{defi} {\rm (Model)} KEDL model {\rm(}denoted as M{\rm)} is $(\Delta^{I},\Sigma^{I},\bullet^{I})$, where $\Delta^{I}$, $\Sigma^{I}$ is respective non-empty domain, $\bullet^{I}$ is interpret function. $\bullet^{I}$ maps the object knowledge element instance $c$ to $c^{I}$ and $c^{I}\in \Delta^{I}$, maps the object knowledge element concept $C$ to $C^{I}$ and $C^{I}\subseteq \Delta^{I}$, maps the relationship $p$ between object knowledge elements to $p^{I}$ and $p^{I}\in \Delta^{I} \times \Delta^{I}$, maps the attribute knowledge element instance $a$ to $a^{I}$ and $a^{I}\in \Sigma^{I}$, maps attribute knowledge element concept $A$ to $A^{I}$ and $A^{I}\subseteq \Sigma^{I}$, maps the relationship $q$ between attribute knowledge elements to $q^{I}$ and $q^{I}\in \Sigma^{I}\times \Sigma^{I}$, maps the relationship $r$ to $r^{I}$ and $r^{I}\in \Delta^{I} \times \Sigma^{I}$, and satisfy the condition: for any relationship $r$ and any $x\in \Delta^{I}$, ${y:(x,y)\in r^{I}}$ contains only one element.
\end{defi}

After defining the description logic KEDL model, the semantic interpretation of description logic KEDL is given below.
\begin{defi} {\rm(Semantics)} The semantics of the description logic KEDL is represented as follows:\newline
{\rm(1)} $\top ^{I}=\Delta^{I}=\{x|x\in(C \sqcup \neg C)^{I}\}=\{x|x\in C^{I} \cup (\neg C)^{I}\}=\{x|x\in C^{I}$ or$\  x\notin C^{I}\}$.\newline
{\rm(2)} $\bot ^{I}=\emptyset=\{x|x\in(C \sqcap \neg C)^{I}\}=\{x|x\in C^{I} \cap (\neg C)^{I}\}=\{x|x\in C^{I}$ and$\ x\notin C^{I}\}$.\newline
{\rm(3)} $M\models C(c),iff\ c^{I} \in C^{I}$.\newline
{\rm(4)} $(\neg C)I=\Delta^{I}\setminus C^{I}=\{x|x\in \Delta^{I}$ and $\ x\notin C^{I}\}$.\newline
{\rm(5)} $(C\sqcup D)^{I}=C^{I}\cup D^{I}=\{x\in\Delta^{I}|x\in C^{I}$ or $\ x\in D^{I}\}$.\newline
{\rm(6)} $(C\sqcap D)^{I}=C^{I}\cap D^{I}=\{x\in\Delta^{I}|x\in C^{I}$ and $\ x\in D^{I}\}$.\newline
{\rm(7)} $(\forall p.C)^{I}=\{x\in\Delta^{I}|\forall y\in\Delta^{I}((x,y)\in p^{I}\Rightarrow y\in C^{I})\}=\{x\in\Delta^{I}|\forall y$\ if$\ (x,y)\in p^{I}$ then$\ y\in C^{I}\}$.\newline
{\rm(8)} $(\exists p.C)^{I}=\{x\in\Delta^{I}|\exists y \in \Delta^{I} ((x,y)\in p^{I} \& y\in C^{I})\}=\{x\in\Delta^{I}|\exists y\ makes\ (x,y)\in p^{I} and\ y\in C^{I}\}$.\newline
{\rm(9)} $(\neg A)I=\Sigma^{I}\setminus A^{I}=\{x|x\in \Sigma^{I} and\ x\notin A^{I}\}$.\newline
{\rm(10)} $(A\sqcup B)^{I}=A^{I}\cup B^{I}=\{u\in \Sigma^{I}|u\in A^{I} or\ u\in B^{I}\}$.\newline
{\rm(11)} $(A\sqcap B)^{I}=A^{I}\cap B^{I}=\{u\in \Sigma^{I}|u\in A^{I} and\ u\in B^{I}\}$.\newline
{\rm(12)} $(\forall q.A)^{I}=\{u\in \Sigma^{I}|\forall v \in\Sigma^{I} ((u,v)\in q^{I}\Rightarrow v\in A^{I})\}=\{u\in \Sigma^{I}|\forall v\ if\ (u,v)\in q^{I} then v\in A^{I}\}$.\newline
{\rm(13)} $(\exists q.A)^{I}=\{u \in \Sigma^{I}|\exists v \in \Sigma^{I} ((u,v)\in q^{I}\& v\in A^{I})\}=\{u\in \Sigma^{I}|\exists v\ makes\ (u,v)\in q^{I} and\ v\in A^{I}\}$.\newline
{\rm(14)} $(r-)^{I} =\{(u,x)|(x,u)\in r^{I}\}$.\newline
{\rm(15)} $(\forall r.A)^{I}=\{x\in \Delta^{I}|\forall u\ \in\Sigma^{I} ((x,u)\in r^{I}\Rightarrow u\in A^{I})\}=\{x\in \Delta^{I}|\forall u\ if\ (x,u)\in r^{I} then\ u\in A^{I}\}$.\newline
{\rm(16)} $(\exists r.A)^{I}=\{x \in \Delta^{I}|\exists u \in \Sigma^{I} ((x,u)\in r^{I}\& u\in A^{I})\}=\{x\in \Delta^{I}|\exists u\ makes\ (x,u)\in r^{I} and\ u\in A^{I}\}$.\newline
{\rm(17)} $(\forall r-.C)^{I}=\{u\in \Sigma^{I}|\forall x\ \in \Delta^{I}((u,x)\in (r-)^{I}\Rightarrow x\in C^{I})\}=\{u\in \Sigma^{I}|\forall x\ \in \Delta^{I} if (u,x)\in (r-)^{I} then\ x\in C^{I}\}$.\newline
{\rm(18)} $(\exists r-.C)^{I}=\{u \in \Sigma^{I}|\exists x \in \Delta^{I} ((u,x)\in (r-)^{I}\& c\in C^{I})\}=\{u\in \Sigma^{I} |\exists x\in \Delta^{I}\ makes\ (u,x)\in (r-)^{I} and\ x\in C^{I}\}$.\newline
{\rm(19)} $M\models r(c,d)$, iff $(c^{I},d^{I})\in r^{I}$.\newline
{\rm(20)} $M\models A(a)$, iff $a^{I}\in A^{I}$.\newline
{\rm(21)} $M\models p(c,d)$, iff $(c^{I},d^{I})\in p^{I}$.\newline
{\rm(22)} $M\models q(c,d)$, iff $(c^{I},d^{I})\in q^{I}$.\newline
{\rm(23)} $M\models C\rightarrow D$, iff $\exists x \in \Delta^{I}$, if$\ x \in C^{I}$ then$\ x\in D^{I}$.\newline
{\rm(24)} $M\models C\leftrightarrow D$, iff $\exists x \in \Delta^{I}$, if$\ x \in C^{I}$ then$\ x\in D^{I}$, and if $x\in D^{I} then\ x\in C^{I}$.\newline
{\rm(25)} $M\models A\rightarrow B$, iff $\exists u \in \Sigma^{I}$, $if\ u \in A^{I}$ then$\ u\in B^{I}$.\newline
{\rm(26)} $M\models A\leftrightarrow B$, iff $\exists u \in \Sigma^{I}$, if$\ u \in A^{I}$ then $\ u\in B^{I}$, and if $u \in B^{I}$ then$\ u\in A^{I}$.
\end{defi}

\subsection{The formal axioms of KEDL}
To further explore the description logic KEDL properties, we will study KEDL from an axiom perspective. The set of formulas described by the description logic KEDL is extended by the following axioms.\newline
{\bf Axiom 1.} $\vdash \neg\phi\sqcup(\neg\psi\sqcup\phi)$. \newline
Equivalent representation $\vdash \phi \rightarrow(\psi\rightarrow\phi)$.\newline
{\bf Axiom 2.} $\vdash \neg(\neg\phi\sqcup(\neg\psi\sqcup\gamma))\sqcup(\neg(\neg\phi\sqcup\psi)\sqcup(\neg\phi\sqcup\gamma))$.\newline
Equivalent representation $\vdash(\phi\rightarrow(\psi\rightarrow\gamma))\rightarrow((\phi\rightarrow\psi)\rightarrow((\phi\rightarrow\gamma ))$.\newline
{\bf Axiom 3.} $\vdash\neg(\phi\sqcup\neg\psi)\sqcup(\neg\psi\sqcup\phi)$.\newline
Equivalent representation $\vdash(\neg\phi\rightarrow\neg\psi)\rightarrow(\psi\rightarrow\phi)$.\newline
{\bf Axiom 4.} $\vdash\neg(\exists p.C\sqcup\exists p.D)\sqcup \exists p.(C\sqcup D)$.\newline
Equivalent representation $\vdash\exists p.C\sqcup\exists p.D\rightarrow\exists p.(C\sqcup D)$.\newline
{\bf Axiom 5.} $\vdash\neg\exists p.(\neg C\sqcup\neg D)\sqcup \neg(\neg\exists p.C\sqcup\neg\exists p.D)$.\newline
Equivalent representation $\vdash \exists p.(C\sqcap D)\rightarrow \exists p.C\sqcap\exists p.D$.\newline
{\bf Axiom 6.} $\vdash(\neg\exists p.C\sqcup\neg\exists p.\neg D)\sqcup\exists p. \neg(\neg C\sqcup\neg D)$.\newline
Equivalent representation $\vdash(\exists p.C\sqcap \forall p.D)\rightarrow\exists p.(C\sqcap D)$.\newline
{\bf Axiom 7.} $\vdash\neg(\exists q.A\sqcup\exists q.B)\sqcup \exists q.(A\sqcup B)$.\newline
Equivalent representation $\vdash\exists q.A\sqcup\exists q.B\rightarrow\exists q.(A\sqcup B)$.\newline
{\bf Axiom 8.} $\vdash\neg\exists q.(\neg A\sqcup\neg B)\sqcup\neg(\neg\exists q.A\sqcup\neg\exists q.B)$.\newline
Equivalent representation $\vdash\exists q.(A\sqcap B)\rightarrow\exists q.A\sqcap\exists q.B$.\newline
{\bf Axiom 9.} $\vdash(\neg\exists q.A\sqcup\exists q.\neg B)\sqcup\exists q.\neg(\neg A\sqcup\neg B)$.\newline
Equivalent representation $\vdash(\exists q.A\sqcap\forall q.B)\rightarrow\exists q.(A\sqcap B)$.\newline
{\bf Axiom 10.} $\vdash\neg(\exists r.A\sqcup\exists r.B)\sqcup\exists r.(A\sqcup B)$.\newline
Equivalent representation $\vdash\exists r.A\sqcup\exists r.B\rightarrow\exists r.(A\sqcup B)$.\newline
{\bf Axiom 11.} $\vdash\neg \exists r.(\neg A\sqcup\neg B)\sqcup\neg(\neg\exists r.A\sqcup\neg\exists r.B)$.\newline
Equivalent representation $\vdash\exists r.(A\sqcap B)\rightarrow\exists r.A\sqcap\exists r.B$.\newline
{\bf Axiom 12.} $\vdash(\neg\exists r.A\sqcup\exists r.\neg B)\sqcup\exists r.\neg(\neg A\sqcup\neg B)$.\newline
Equivalent representation $\vdash(\exists r.A\sqcap\forall r.B)\rightarrow\exists r.(A\sqcap B)$.\newline
{\bf Axiom 13.} $\vdash\neg(\exists r^{-}.C\sqcup\exists r^{-}.D)\sqcup\exists r^{-}.(C\sqcup D)$.\newline
Equivalent representation $\vdash\exists r^{-}.C\sqcup\exists r^{-}.D\rightarrow\exists r^{-}.(C\sqcup D)$.\newline
{\bf Axiom 14.} $\vdash\neg \exists r^{-}.(\neg C\sqcup\neg D)\sqcup\neg(\neg\exists r^{-}.C\sqcup\neg\exists r^{-}.D)$.\newline
Equivalent representation $\vdash\exists r^{-}.(C\sqcap D)\rightarrow\exists r^{-}.C\sqcap\exists r^{-}.D$.\newline
{\bf Axiom 15.} $\vdash(\neg\exists r^{-}.C\sqcup\exists r^{-}.\neg D)\sqcup\exists r^{-}.\neg(\neg C\sqcup\neg D)$.\newline
Equivalent representation $\vdash(\exists r^{-}.C\sqcap\forall r^{-}.D)\rightarrow\exists r^{-}.(C\sqcap D)$.\newline
{\bf Axiom 16.} $\vdash\neg\exists r^{-}.\forall r.A\sqcup A$.\newline
Equivalent representation $\vdash\exists r^{-}.\forall r.A\rightarrow A$.\newline
{\bf Axiom 17.} $\vdash\neg\exists r.\forall r^{-}.C\sqcup C$.\newline
Equivalent representation $\vdash\exists r.\forall r^{-}.C\rightarrow C$.\newline
{\bf Axiom 18.} If\ $\vdash\phi(c)$,and\ $\vdash \phi\rightarrow\psi$,then\ $\vdash\psi(c)$.\newline
{\bf Axiom 19.} If\ $\vdash\phi\rightarrow\psi$,and\ $\vdash\psi\rightarrow\phi$,then\ $\vdash\phi\leftrightarrow\psi$.\newline
{\bf Axiom 20.} If\ $\vdash\phi\rightarrow\psi$,and\ $\vdash\psi\rightarrow\gamma$,then\ $\vdash\phi\rightarrow\gamma$.\newline
{\bf Axiom 21.} $\vdash\phi\rightarrow\psi\sqcap\gamma$,iff $\vdash\phi\rightarrow\psi$ and\ $\vdash\phi\rightarrow\gamma$.\newline
Here $\phi$, $\psi$ and $\gamma$ are the object knowledge element concepts or the attribute knowledge element concepts at the same time. $C$, $D$ are the concepts of object knowledge element. $A$,$B$ are the concepts of attribute knowledge element. $p$ is the relationship between object knowledge elements. $q$ is the relationship between attribute knowledge elements. $r$ is the relationship between object knowledge element and attribute knowledge element. The specific explanation of the KEDL axiom is stated as follows (among these axioms, Axiom 1, 2, and 3 are axioms in predicate logic):\newline
{\bf Axiom 4.} $\vdash\neg(\exists p.C\sqcup\exists p.D)\sqcup \exists p.(C\sqcup D)$ denotes in domain $\Delta^{I}$. If there is an object knowledge element concept, the instance individual of it has a relationship $p$ with the instance individual of the object knowledge element concept $C$, or the instance individual of the object knowledge element concept has relationship $p$ with the instance individual of the object knowledge element concept $D$, then there is an object knowledge element concept, the instance individual of it has relationship $p$ with the instance individual of the object knowledge element concept $C$ or $D$.\newline
{\bf Axiom 5.} $\vdash\neg\exists p.(\neg C\sqcup\neg D)\sqcup \neg(\neg\exists p.C\sqcup\neg\exists p.D)$ denotes in domain $\Delta^{I}$. If there is an object knowledge element concept and its instance individual has a relationship $p$ with the same instance individual of the object knowledge element concepts $C$ and $D$, then there is an object knowledge element concept whose instance individual has relationship $p$ with the instance individual of $C$, and the instance individual of the object knowledge element concept has relationship $p$ with the instance individual of the object knowledge element concept $D$.\newline
{\bf Axiom 6.} $\vdash(\neg\exists p.C\sqcup\neg\exists p.\neg D)\sqcup\exists p. \neg(\neg C\sqcup\neg D)$ denotes in domain $\Delta^{I}$. If there is an object knowledge element concept,its instance individual has a relationship $p$ with the instance individual of the object knowledge element concept $C$, and the object knowledge element concept whose instance individual has relationship $p$ with some instance individuals (these instance individuals belong to the object knowledge element concept $D$),then the object knowledge element concept has relationship $p$ with the same instance individual of the object knowledge element concepts $C$ and $D$.\newline
{\bf Axiom 7.} $\vdash\neg(\exists q.A\sqcup\exists q.B)\sqcup \exists q.(A\sqcup B)$ denotes in domain $\Sigma^{I}$. If there is an attribute knowledge element concept, its instance individual has a relationship $q$ with the instance individual of the attribute knowledge element concept $A$. Or if there is an attribute knowledge element concept, and its instance individual has a relationship $q$ with the attribute knowledge element concept $B$, then there is an attribute knowledge element concept whose instance individual has a relationship $q$ with the instance individual of the attribute knowledge element concept $A$ or $B$.\newline
{\bf Axiom 8.} $\vdash\neg\exists q.(\neg A\sqcup\neg B)\sqcup\neg(\neg\exists q.A\sqcup\neg\exists q.B)$ denotes in domain $\Sigma^{I}$. If the attribute knowledge element concept has a relationship $q$ with the same instance individual of the attribute knowledge element concept $A$ and $B$, then there is an attribute knowledge element concept whose instance individual has relationship $q$ with the instance individual of the attribute knowledge element concept $A$. There is also an attribute knowledge element concept whose instance individual has relationship $q$ with the instance individual of the attribute knowledge element $B$.\newline
{\bf Axiom 9.} $\vdash(\neg\exists q.A\sqcup\exists q.\neg B)\sqcup\exists q.\neg(\neg A\sqcup\neg B)$ denotes in domain $\Sigma^{I}$. If there is an attribute knowledge element concept, its instance individual has a relationship $q$ with the instance individual of the attribute knowledge element concept $A$, and there is an attribute knowledge element concept, its instance individual has relationship $q$ with any instance individuals (these instance individuals belong to the attribute knowledge element concept $B$), then there is an attribute knowledge element concept whose instance individual has relationship $q$ with the same instance individual of the attribute knowledge element concept A and $B$.\newline
{\bf Axiom 10.} $\vdash\neg(\exists r.A\sqcup\exists r.B)\sqcup\exists r.(A\sqcup B)$ denotes in the domains $\Delta^{I}$ and $\Sigma^{I}$. If there is an object knowledge element concept whose instance individual has a relationship $r$ with the instance individual of the attribute knowledge element concept $A$ or there is an object knowledge element concept whose instance individual has relationship $r$ with the instance individual of the attribute knowledge element concept $B$, then the instance individual of the object knowledge element concept has a relationship $r$ with the instance individual of the attribute knowledge element concept $A$ or the instance individual of the attribute knowledge element concept $B$.\newline
{\bf Axiom 11.} $\vdash\neg \exists r.(\neg A\sqcup\neg B)\sqcup\neg(\neg\exists r.A\sqcup\neg\exists r.B)$ denotes in the domains $\Delta^{I}$ and $\Sigma^{I}$. If the object knowledge element concept whose instance individual has relationship $r$ with the same instance individual of the attribute knowledge element concepts $A$ and $B$, then the object knowledge element concept whose instance individual has relationship $r$ with the instance individual of the attribute knowledge element concepts $A$, and the object knowledge element concept whose instance individual has relationship $r$ with the instance individual of the attribute knowledge element concepts $B$. \newline
{\bf Axiom 12.} $\vdash(\neg\exists r.A\sqcup\exists r.\neg B)\sqcup\exists r.\neg(\neg A\sqcup\neg B)$ denotes in the domains $\Delta^{I}$ and $\Sigma^{I}$. If there is an object knowledge element concept whose instance individual has a relationship $r$ with the instance individual of the attribute knowledge element concept $A$, and there is an object knowledge element concept whose instance individuals have relationship $r$ with any instance individuals (these instance individuals are all instance individuals of the attribute knowledge element concept $B$),then there is an object knowledge element concept whose instance individual has relationship $r$ with the same instance individual of the attribute knowledge element concepts $A$ and $B$.\newline
{\bf Axiom 13.} $\vdash\neg(\exists r^{-}.C\sqcup\exists r^{-}.D)\sqcup\exists r^{-}.(C\sqcup D)$ denotes in the domains $\Delta^{I}$ and $\Sigma^{I}$. If there is an attribute knowledge element concept whose instance individual has a relationship $r^{-}$ with the instance individual of the object knowledge element concept $C$, or there is an attribute knowledge element concept whose instance individual has a relationship $r^{-}$ with the instance individual of the object knowledge element concept $D$, then there is an attribute knowledge element concept whose instance individual has a relationship $r^{-}$ with the instance individual of the object knowledge element concept $C$ or $D$.\newline
{\bf Axiom 14.} $\vdash\neg \exists r^{-}.(\neg C\sqcup\neg D)\sqcup\neg(\neg\exists r^{-}.C\sqcup\neg\exists r^{-}.D)$ denotes in the domains $\Delta^{I}$ and $\Sigma^{I}$. If there is an attribute knowledge element concept whose instance individual has a relationship $r^{-}$ with the same instance individual of the object knowledge element concepts $C$ and $D$, then there is an attribute knowledge element concept whose instance individual has relationship $r^{-}$ with the instance individual of the object knowledge element concept $C$, and there is also an attribute knowledge element concept whose instance entity has a relationship $r^{-}$ with the instance individual of the object knowledge element concept $D$.\newline
{\bf Axiom 15.} $\vdash(\neg\exists r^{-}.C\sqcup\exists r^{-}.\neg D)\sqcup\exists r^{-}.\neg(\neg C\sqcup\neg D)$ denotes in the domains $\Delta^{I}$ and $\Sigma^{I}$. If there is an attribute knowledge element concept whose instance individual has relationship $r^{-}$ with the instance individual of the object knowledge element concept $C$, and there is an attribute knowledge element concept whose instance individuals have relationship $r^{-}$ with any instance individuals (these instance individuals belong to the object knowledge element concept $D$), then there is an attribute knowledge element concept whose instance individual has relationship $r^{-}$ with the same instance individual of the object knowledge element concept $C$ and $D$.\newline
{\bf Axiom 16.} $\vdash\neg\exists r^{-}.\forall r.A\sqcup A$ denotes in domains $\Delta^{I}$ and $\Sigma^{I}$. If there is an attribute knowledge element concept whose instance individual has relation $r^{-}$ with some instance individuals, and these instance individuals have relation $r$ with any instance individuals that belong to the attribute knowledge element concept $A$, then the attribute knowledge element concept $A$ can be satisfied.\newline
{\bf Axiom 17.} $\vdash\neg\exists r.\forall r^{-}.C\sqcup C$ denotes in domains $\Delta^{I}$ and $\Sigma^{I}$. If there is an object knowledge element concept whose instance individual has relation $r$ with some instance individuals, and these instance individuals have relation $r^{-}$ with any instance individuals that belong to the object knowledge element concept $C$, then the object knowledge element concept $C$ can be satisfied.\newline
{\bf Axiom 18.} If\ $\vdash\phi(c)$, and\ $\vdash \phi\rightarrow\psi$, then\ $\vdash\psi(c)$ denotes in domains $\Delta^{I}$, an instance individual $c$ belongs to the object knowledge element concept $\phi$, and the object knowledge element concept $\phi$ contains the object knowledge element concept $\psi$, then the instance individual $c$ also belongs to the object knowledge element concept $\psi$. Or in domain $\Sigma^{I}$, an instance individual $c$ belongs to the attribute knowledge element concept $\phi$, and the attribute knowledge element concept $\phi$ contains the attribute knowledge element concept $\psi$, then the instance individual $c$ also belongs to the attribute knowledge element concept $\psi$.\newline
{\bf Axiom 19.} If\ $\vdash\phi\rightarrow\psi$, and\ $\vdash\psi\rightarrow\phi$, then\ $\vdash\phi\leftrightarrow\psi$ denotes in domain $\Delta^{I}$, the object knowledge element concept $\phi$ implies the object knowledge element concept $\psi$ and the object knowledge element concept $\psi$ also implies the object knowledge element concept $\phi$, then the object knowledge element concept $\phi$ is equivalent to the object knowledge element concept $\psi$. Or in domain $\Sigma^{I}$, the attribute knowledge element concept $\phi$ implies the attribute knowledge element concept $\psi$ and the attribute knowledge element concept $\psi$ also implies the attribute knowledge element concept $\phi$, then the attribute knowledge element concept $\phi$ is equivalent to the attribute knowledge element concept $\psi$.\newline
{\bf Axiom 20.} If\ $\vdash\phi\rightarrow\psi$,and\ $\vdash\psi\rightarrow\gamma$, then\ $\vdash\phi\rightarrow\gamma$ denotes in domain $\Delta^{I}$, the object knowledge element concept $\phi$ implies the object knowledge element concept $\psi$, and the object knowledge element concept $\psi$ implies the object knowledge element concept $\gamma$,then object knowledge element concept $\phi$ implies the object knowledge element concept $\gamma$. Or in domain $\Sigma^{I}$, the attribute knowledge element concept $\phi$ implies the attribute knowledge element concept $\psi$, and the attribute knowledge element concept $\psi$ implies the attribute knowledge element concept $\gamma$, then attribute knowledge element concept $\phi$ implies the attribute knowledge element concept $\gamma$.\newline
{\bf Axiom 21.} $\vdash\phi\rightarrow\psi\sqcap\gamma$, if $\vdash\phi\rightarrow\psi$ and\ $\vdash\phi\rightarrow\gamma$ denotes in domain $\Delta^{I}$, the object knowledge element concept $\phi$ implies the intersection of the object knowledge element concept $\psi$ and $\gamma$, if and only if the object knowledge element concept $\phi$ implies the object knowledge element concept $\psi$ and the object knowledge element concept $\phi$ implies the object knowledge element concept $\gamma$. Or in domain $\Sigma^{I}$, the attribute knowledge element concept $\phi$ implies the intersection of the attribute knowledge element concept $\psi$ and $\gamma$, if and only if the attribute knowledge element concept $\phi$ implies the attribute knowledge element concept $\psi$ and the attribute knowledge element concept $\phi$ implies the attribute knowledge element concept $\gamma$.\newline

\subsection{The property of KEDL}
By demonstrating, the description logic KEDL system satisfies the deduction theorem, the transposition law,the counter-evidence law and the law of absurdity. It also has the following properties:\newline
{\bf Property 1.} 1) $\phi\sqcap\phi\leftrightarrow\phi$;      2)$\phi\sqcup\phi\leftrightarrow\phi$. (Idempotent law)\newline
{\bf Property 2.} 1)$\phi\sqcap\psi\leftrightarrow\psi\sqcap\phi$;    2)$\phi\sqcup\psi\leftrightarrow\psi\sqcup\phi$.  (Commutation law)\newline
{\bf Property 3.} 1)$(\phi\sqcap\psi)\sqcap\gamma\leftrightarrow\phi\sqcap(\psi\sqcap\gamma)$;
2)$(\phi\sqcup\psi)\sqcup\gamma\leftrightarrow\phi\sqcup(\psi\sqcup\gamma)$.      (Combination law)\newline
{\bf Property 4.} 1)$\phi\sqcup(\psi\sqcap\gamma)\leftrightarrow(\phi\sqcup\psi)\sqcap(\phi\sqcup\gamma)$;
2)$\phi\sqcap(\psi\sqcup\gamma)\leftrightarrow(\phi\sqcap\psi)\sqcup(\phi\sqcap\gamma)$.       (Distribution law)\newline
{\bf Property 5.} 1)$\phi\sqcup\bot\leftrightarrow\phi$;       2)$\phi\sqcap\top\leftrightarrow\phi$.     (Identity law)\newline
{\bf Property 6.} 1)$\phi\sqcup\top\leftrightarrow\top$;                   2)$\phi\sqcap\bot\leftrightarrow\bot$. \newline
{\bf Property 7.} $\neg\phi\sqcup\phi\leftrightarrow\top$.   (Exclude-middle law)\newline
{\bf Property 8.} $\phi\sqcap\neg\phi\leftrightarrow\bot$.   (Contradiction Law)\newline
{\bf Property 9.} 1)$\phi\sqcup(\phi\sqcap\psi)\leftrightarrow\phi$;       2)$\phi\sqcap(\phi\sqcup\psi)\leftrightarrow\phi$.     (Absorption law)\newline
{\bf Property 10.} 1)$\neg(\phi\sqcap\psi)\leftrightarrow\neg\phi\sqcup\neg\psi$; \newline
2)$\neg(\phi\sqcup\psi)\leftrightarrow\neg\phi\sqcap\neg\psi$.    (De. Morgan law) \newline
{\bf Property 11.} 1)$\neg\bot\leftrightarrow\top$;     2)$\neg\top\leftrightarrow\bot$.   (Over-complement law)\newline
{\bf Property 12.} $\neg\neg\phi\leftrightarrow\phi$;     (Double-negative law)\newline

In the following, we give the proof process for the first property of property 1, and the proof process of other properties are not listed one by one because of the length limitations of the paper.\newline

The proof of Property 1 is as follows.
\begin{proof}
To prove $\phi\sqcap\phi\leftrightarrow\phi$, is to prove $\phi\sqcap\phi\rightarrow\phi$ and $\phi\rightarrow\phi\sqcap\phi$.

For the first part, it needs to check $\neg(\phi\rightarrow\neg\phi)\rightarrow\phi$ as follows: \newline
(1) $\neg\phi\rightarrow(\phi\rightarrow\neg\phi)$\ \ \ Axiom 1\newline
(2) $(\neg\phi\rightarrow(\phi\rightarrow\neg\phi))\rightarrow(\neg(\phi\rightarrow\neg\phi)\rightarrow\neg\neg\phi)$\
    Axiom 3\newline
(3) $(\neg(\phi\rightarrow\neg\phi)\rightarrow\neg\neg\phi)$\ (1),(2), MP\newline
(4) $\neg\neg\phi\rightarrow\phi$\ Double-negative law\newline
(5) $(\neg(\phi\rightarrow\neg\phi)\rightarrow\phi)$\ (3),(4), Axiom 20\newline

Next, we show the second part, and it is to prove $\phi\rightarrow\neg(\phi\rightarrow\neg\phi)$. According to the deductive theorem, to prove ${\phi}\vdash\neg(\phi\rightarrow\neg\phi)$, put $\phi\rightarrow\neg\phi$ as a hypothesis, there will be\newline
(1) ${\phi,\phi\rightarrow\neg\phi}\vdash\phi$; \newline
(2) ${\phi,\phi\rightarrow\neg\phi}\vdash\neg\phi$. \newline
By (1), (2) through the law of absurdity, to get ${\phi}\vdash\neg(\phi\rightarrow\neg\phi)$. 
\end{proof}

\begin{thm}\label{result1}
       All axioms of description logic KEDL system are valid.
\end{thm}

We first prove that Axiom 1 $\phi\rightarrow(\psi\rightarrow\phi)$ is valid.\\
\begin{proof}
We only prove its equivalent formula $\neg\phi\sqcup(\neg\psi\sqcup\phi)$ as an effective formula.
For any explanation $\bullet^{I}$, there is two cases as follows.\newline
(1) If the formula for  the object knowledge element concept, then
\begin{eqnarray*}
&\quad&(\neg\phi\sqcup(\neg\psi\sqcup\phi))^{I}\\
&=&(\neg\phi)^{I}\cup(\neg\psi\sqcup\phi)^{I}\\
&=&(\Delta^{I}\setminus\phi^{I})\cup(\neg\psi)^{I}\cup\phi^{I}\\
&=&(\Delta^{I}\setminus\phi^{I})\cup(\Delta^{I}\setminus\psi^{I})\cup\phi^{I}\\
&=&\Delta^{I}\cup(\Delta^{I}\setminus\psi^{I})\\
&=&\Delta^{I}.
\end{eqnarray*}
(2) If the formula for the attribute knowledge element concept, then
\begin{eqnarray*}
&\quad&(\neg\phi\sqcup(\neg\psi\sqcup\phi))^{I}\\
&=&(\neg\phi)^{I}\cup(\neg\psi\sqcup\phi)^{I}\\
&=&(\Sigma^{I}\setminus\phi^{I})\cup(\neg\psi)^{I}\cup\phi^{I}\\
&=&(\Sigma^{I}\setminus\phi^{I})\cup(\Sigma^{I}\setminus\psi^{I})\cup\phi^{I}\\
&=&\Sigma^{I}\cup(\Sigma^{I}\setminus\psi^{I})\\
&=&\Sigma^{I}.
\end{eqnarray*}
Since $(\neg\phi\sqcup(\neg\psi\sqcup\phi))^{I}$ is all valid in the above situation,
we infer $\models\neg\phi\sqcup(\neg\psi\sqcup\phi)$.
It implies $\models\phi\rightarrow(\psi\rightarrow\phi)$.
\end{proof}

Axiom2-Axiom21 are the effective formula by the proof, and the proof process are not listed one by one because of the length limitations of the paper.

Next, we discuss the reliability theorem of the description logic system KEDL, which reveals that the grammar provability can be used to derive the semantic provability.

\begin{prop}\label{tt1}
$\Gamma\models\phi$ and $\Gamma\models\phi\rightarrow\psi$, then $\Gamma\models\psi$.
\end{prop}
\begin{proof}
$\Gamma\models\phi$ implies for any explanation $\bullet^{I}$. (1) $\phi^{I}\subseteq\Delta^{I}$ is established. Using $\Gamma\models\phi\rightarrow\psi$, we have $\phi^{I}\supseteq\psi^{I}$. Thus    $\psi^{I}\subseteq\Delta^{I}$ is established. (2) $\phi^{I}\subseteq\Sigma^{I}$ is established, we deduce $\Gamma\models\phi\rightarrow\psi$ and $\phi^{I}\supseteq\psi^{I}$. Hence, $\psi^{I}\subseteq\Sigma^{I}$. Therefore, $\Gamma\models\psi$.
\end{proof}

\begin{thm}\label{theorem2}
        (Reliability of KEDL) $\Gamma\vdash\phi\Rightarrow\Gamma\models\phi$.
\end{thm}
\begin{proof}
Let $\Gamma\vdash\phi$, then there is a proof of $\phi$ from $\Gamma$, that is $\phi_{1}$, $\phi_{2}$, $\cdots$, $\phi_{n}$, and $\phi_{n}=\phi$.

Next, to prove $\Gamma\models\phi$, we induce the length $n$ of the $\phi$'s proof.

When $n=1$, $\phi_{1}=\phi$, there are two cases at this time. (1)$\phi$ is an axiom of description logic KEDL. By theorem 1, we know that $\phi$ is a valid formula. (2) $\phi\in\Gamma$, by the definition of the semantic inference, then  $\Gamma\models\phi$.

If $n>1$, there are three cases. (1) $\phi$ is an axiom of description logic KEDL; (2) $\phi\in\Gamma$. These two cases are the same as the situations when $n=1$, $\Gamma\models\phi$ are all established.

Case(3): $\phi$ is derived by modus ponens, that is, there are $i, j$, and $i, j<n$, then $\phi_{j}=\phi_{i}\rightarrow\phi$. This time, there are $\Gamma\models\phi_{i}$ and $\Gamma\models\phi_{j}$ from $\Gamma\vdash\phi_{i}$ and $\Gamma\vdash\phi_{j}$ by inductive inference. The latter is $\Gamma\models\phi_{i}\rightarrow\phi$. Then $\Gamma\models\phi$ is established by proposition \ref{tt1}.
\end{proof}

\begin{lem}\label{lemma1}
        L(X) of KEDL is a countable set.
\end{lem}
\begin{proof}
In description logic KEDL, some constructors, such as $\neg$, $\sqcap$, $\sqcup$, $\forall p.C$, $\exists p.C$, $\forall q.A$, $\exists q.A$, $\forall r.A$, $\exists r.A$, $\forall r^{-}.C$, $\exists r^{-}.C$, are introduced. Supposing the set which is constructed by these KEDL constructors is as follows $X'$ and $X$.
$$X'={ \neg, \sqcap, \sqcup, \forall p.C, \exists p.C, \forall q.A, \exists q.A, \forall r.A, \exists r.A, \forall r^{-}.C, \exists r^{-}.C},$$
$$X={C_{1},C_{2},\cdots} \cup {A_{1},A_{2},\cdots}\cup {p_{1},p_{2},\cdots}\cup {q_{1},q_{2},\cdots}\cup {r_{1},r_{2},\cdots}.$$
Then, construct $L_{0}$, $L_{1}$, $L_{2}$, $\cdots$, from $X'$ and $X$. The Constructed method is as follows.\newline
$$L_{0}={C_{1},C_{2},\cdots}\cup {A_{1},A_{2},\cdots}\cup {p_{1},p_{2},\cdots}\cup {q_{1},q_{2},\cdots}\cup {r_{1},r_{2},\cdots},$$
$L_{1}=\{\neg C_{1}, \neg C_{2}, \cdots, C_{1}\sqcup C_{2}, C_{2}\sqcup C_{1}, \cdots, C_{1}\sqcap C_{2}, C_{2}\sqcap C_{1}, \cdots, \neg A_{1}, \neg A_{2}, \cdots,A_{1} \sqcup A_{2}, A_{2} \sqcup A_{1},\cdots, A_{1}\sqcap A_{2}, A_{2}\sqcap A_{1}, \cdots, \forall p.C_{1}, \forall p.C_{2}, \cdots, \exists p.C_{1}, \exists p.C_{2}, \cdots, \forall q.A_{1}, \forall q.A_{2}, \cdots,\exists q.A_{1},\\
\exists q.A_{2}, \cdots, \forall r.A_{1}, \forall r.A_{2}, \cdots, \exists r.A_{1}, \exists r.A_{2}, \cdots, \forall r^{-}.C_{1}, \forall r^{-}.C_{2}, \cdots,\exists r^{-}.C_{1}, \exists r^{-}.C_{2}, \cdots\},$\\
$L_{k}=\{\neg C_{k-1},\cdots,C_{1}\sqcup C_{k-1},\cdots,C_{1}\sqcap C_{k-1},\cdots,\neg A_{k-1},\cdots,A_{1}\sqcup A_{k-1},\cdots,A_{1}\sqcap A_{k-1},\cdots,$\\
$\forall p.C_{k-1},\cdots,\exists p.C_{k-1},\cdots,\forall q.A_{k-1},\cdots,\exists q.A_{k-1},\cdots,\forall r.A_{k-1},\cdots,\exists r.A_{k-1},\\
\cdots,\forall r^{-}.C_{k-1},\cdots, \exists r^{-}.C_{k-1},\cdots\}$ (where $k>0$),\\
$$L(X)=\bigcup^{\infty}_{k=0}L_{k}.$$
$L(X)$ is a countable set and we provide the desired result.
\end{proof}

\begin{thm}\label{theorem3}
        (Completeness of KEDL) $\Gamma\models\phi\Rightarrow\Gamma\vdash\phi$.
\end{thm}
\begin{proof}
We assume $\Gamma\vdash\phi$ is not established. To construct or find an explanation $\bullet^{I}$ which makes all the formulas of $\Gamma$ are $\Delta^{I}$ or $\Sigma^{I}$, which contradicts to $\Gamma\models\phi$ due to $\phi=\emptyset$.

Since $L(X)$ is a countable set, arrange all formulas of KEDL together, assuming it is $\phi_{1}$, $\phi_{2}$, $\cdots$, $\phi_{n}$, $\cdots$. Let $\Gamma_{0}=\Gamma\cup\{\neg\phi\}$.
If $n>0$, let $$\Gamma_{n}=
\begin{cases}
\Gamma_{n-1}& \text{if $\Gamma_{n-1}\vdash\phi_{n-1}$,}\\
\Gamma_{n-1}\cup\{\neg\phi_{n-1}\}& \text{if $\Gamma_{n-1}\vdash\phi_{n-1}$ is invalid.}
\end{cases}$$
Then there is a sequence $\Gamma_{n}$:$\Gamma_{0}\subseteq\Gamma_{1}\subseteq\Gamma_{2}\subseteq\cdots$.\newline
(1) To prove that there is no contradiction to the sequence $\Gamma_{0}$, $\Gamma_{1}$, $\Gamma_{2}$, $\cdots$.
The following is a summary of $n$ to prove that each $\Gamma_{n}$ is no contradiction.

When $k=0$, $\Gamma_{0}\cup{\neg\phi}$ is no contradiction. Otherwise, by $\Gamma_{0}\cup{\neg\phi}\vdash\psi$ and $\neg\psi$, through the law of counter-evidence, $\Gamma\vdash\phi$ is got. But it has been assumed that $\Gamma\vdash\phi$ is not established.

Suppose when $k=n-1$, $\Gamma_{n-1}$ is no contradiction. Next, we consider when $k=n$, to prove that $\Gamma_{n}$ is true. Assuming that $\Gamma_{n}$ is contradictory, then there is $\phi$ and it let there are as follows.\newline
\textcircled{1} $\Gamma_{n}\vdash\phi,\neg\phi$.
In this case, we yield $\Gamma_{n}\neq\Gamma_{n-1}$. By means of the definition of $\Gamma_{n}$, we know the following facts.\newline
\textcircled{2} $\Gamma_{n-1}\vdash\phi_{n-1}$ is invalid.\newline
\textcircled{3} $\Gamma_{n}=\Gamma_{n-1}\cup{\neg\phi_{n-1}}$.\newline
In terms of \textcircled{1}, \textcircled{3}, and the law of counter-evidence, we get $\Gamma_{n-1}\vdash\phi_{n-1}$, which is contradicts to formula \textcircled{2}. Therefore, each $\Gamma_{n}$ has no contradiction.\newline
(2) To construct $\Gamma^{\ast}=\bigcup^{\infty}_{n=0}\Gamma_{n}$, then $\Gamma^{\ast}$ is not contrary.\newline
Assuming that $\Gamma^{\ast}$ is contrary, then according to $\Gamma^{\ast}\vdash\psi,\neg\psi$, we can get conclusion as follows. There is some sufficiently large $n$, $\Gamma_{n}\vdash\psi$ and $\neg\psi$. It's contrary to $\Gamma_{n}$ is not contrary.\newline
(3) $\Gamma^{\ast}$ is complete. That is, for any $\psi$, one of $\Gamma^{\ast}\vdash\psi$ and $\Gamma^{\ast}\vdash\neg\psi$ must be in existence.\newline
Assume $\psi=\phi_{n}$, because $\psi$ must appear in $\phi_{1}$, $\phi_{2}$, $\cdots$. If $\Gamma^{\ast}\vdash\phi_{n}$ is invalid, then $\Gamma_{n}\vdash\phi_{n}$ is invalid. By the definition of $\Gamma_{n}$, there is $\Gamma_{n+1}=\Gamma_{n}\cup\{\neg\phi_{n}\}$. So there is  $\Gamma^{\ast}\vdash\neg\phi_{n}$. This means that for any $n$, one of $\Gamma^{\ast}\vdash\phi_{n}$ and $\Gamma^{\ast}\vdash\neg\phi_{n}$ must be existed.\newline
(4) To construct interpretation $\bullet^{I}$, define a map:
$$\Gamma(\psi)=
\begin{cases}
\psi^{I}=\Delta^{I}& \text{if $\Gamma^{\ast}\vdash\psi$,}\\
(\neg\phi)^{I}=\Delta^{I}& \text{if $\Gamma^{\ast}\vdash\neg\psi$.}
\end{cases}$$
or
$$\Gamma(\psi)=
\begin{cases}
\psi^{I}=\Sigma^{I}& \text{if $\Gamma^{\ast}\vdash\psi$,}\\
(\neg\phi)^{I}=\Sigma^{I}& \text{if $\Gamma^{\ast}\vdash\neg\psi$.}
\end{cases}$$
Since $\Gamma^{\ast}$ is complete and $\Gamma^{\ast}$ is non-contradictory, this definition is reasonable.

For explanation $\bullet^{I}$, $\forall\psi\in\Gamma$, $\psi\in\Gamma\Rightarrow\psi\in\Gamma^{\ast}\Rightarrow\Gamma^{\ast}\vdash\psi$, so $\psi^{I}=\Delta^{I}$(or $\Sigma^{I}$). But $\neg\phi\in\Gamma_{0}\subseteq\Gamma^{\ast}$, then $\Gamma^{\ast}\vdash\neg\phi$, so $(\neg\phi)^{I}=\Delta^{I}$(or $\Sigma^{I}$), that is to say, $\Delta^{I}\setminus\phi^{I}=\Delta^{I}$(or $\Sigma^{I}\setminus\phi^{I}=\Sigma^{I}$), therefore $\phi^{I}=\emptyset$. In this way, in the structured interpretation $\bullet^{I}$, all the formulas are interpreted as $\Delta^{I}$(or $\Sigma^{I}$), but only $\phi^{I}$ is interpreted as $\emptyset$. That is to say, $\Gamma\models\phi$ is not established.
\end{proof}

\section{An example of describing knowledge element by description logic KEDL}
The following example uses the description logic KEDL to describe the knowledge elements which are in \cite{30}. The knowledge element are as follows.\newline
$\bullet$\ Gas $\{$gas composition, ignition point, temperature, gas concentration, gas volume$\}$.\newline
$\bullet$\ Fire-source$\{$location, fire-source category, fire-source temperature$\}$.\newline
$\bullet$\ Gas-explosion$\{$time, location, gas concentration, fire source category, explosion impact, explosion energy$\}$.\newline
$\bullet$\ Tunnel$\{$location, length, width, height, anti-explosive impact, explosion impact$\}$.\newline

The above knowledge element can be described by description logic KEDL as follows.\\
$\bullet$\ The object knowledge element concept include gas, fire-source, gas explosion, roadway.\newline
$\bullet$\ The attribute knowledge element concept include gas composition, ignition point, temperature, gas concentration, gas volume, location, fire source category, fire source temperature, time, explosion impact force, explosion energy, length, width, height, anti-explosion impact force.\newline
$\bullet$\ The relationship between object knowledge element and attribute knowledge element include has-composite, has-fire spot, has-temperature, has-gas density, has-gas amount, has-location, has-fire kind, has-fire temperature, has-time, has-blast impact power, has-blast energy, has-length, has-width, has-height, has-disblast impact power.\newline

Description logic KEDL is used to describe the knowledge elements as follows.\newline
$\bullet$\ Gas$=\exists$has-composite. Gas composition $\sqcap$ has-fire spot. Fire point $\sqcap$ has-temperature.temperature $\sqcap$ has- gas density. Gas concentration $\sqcap$ has-gasamount. Gas volume.\newline
$\bullet$\ Fire-source$=\exists$has-location. Location $\sqcap$ has fire kind. Fire source category $\sqcap$ has-fire temperature. Source temperature.\newline
$\bullet$\ Gas explosion$=\exists$has-time. Time $\sqcap$ has-location. Location $\sqcap$ has-gas density. Gas concentration $\sqcap$ has fire kind. Fire source category $\sqcap$ has-blast impact power. Explosive impact $\sqcap$ has-blast energy. Explosive energy.\newline
$\bullet$\ Tunnel$=\exists$has-location. Location $\sqcap$ has-length. Length $\sqcap$ has-width. Width $\sqcap$ has-height. Height $\sqcap$ has-disblast impact power. Anti-explosive impact $\sqcap$ has-blast impact power. Explosive impact.

Then the knowledge element described by description logic KEDL is edited and inferred by prot¨¦g¨¦ based on Reasonpro1.9 (a reasoner based on description logic) and displayed by graphical plug-in. And the implicit relationship is displayed where describing the knowledge element, which is shown in figure 1.
\begin{figure}[h]
\centering
\includegraphics[width=.8\textwidth]{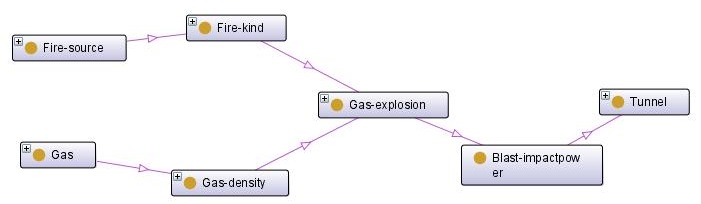}
\caption{KEDL showing the recessive relation between knowledge element}
\end{figure}

In addition, the tunnel with a length greater than 1200 meters cannot be expressed under the traditional description logic ALC, and it can be expressed as has-length.(more-than. {1200 meters}) under the description logic KEDL, where has-length indicates that the tunnel has the length attribute, {1200 meters} is a concept consisting of the instances whose length are 1200 meters.¡¢

From the above example, the description logic KEDL has a stronger description capability than the traditional description logic ALC. It can clearly express the two kinds of knowledge element and the implicit relationship among the knowledge elements.

\section{Conclusion}
In view of the fact that the existing knowledge element representation method is limited in functions to deal with the formality, logic and reasoning. Based on description logic ALC and the common knowledge element model, in order to describe the knowledge element, this paper extended description logic ALC. The concept is extended to two (that is, object knowledge element concept and attribute knowledge element concept),and the relationship is extended to three (that is, the relationship between object knowledge element concept and attribute knowledge element concept, the relationship among object knowledge element concepts, and the relationship among attribute knowledge element concepts), and the inverse relationship constructor is added to propose a description logic KEDL system. The description logic system KEDEL is proved to have reliability, completeness, etc. Finally, It is verified by an example that the description logic KEDL system has strong knowledge element description ability. In addition, logical reasoning based on the description logic KEDL is a problem to be studied in the future.\newline

\section*{Acknowledgments}  The research is supported by National Nature Science Fund Project (No. 61661051 and 61562093), Key Project of Applied Basic Research Program of Yunnan Province (No. 2016FA024).

\end{document}